\newcommand{\longerversion}[1]{}
\newcommand{\longversion}[1]{}
\newcommand{\shortversion}[1]{#1}
\shortversion{
\NeedsTeXFormat{LaTeX2e}
\documentclass{new_tlp}
\usepackage{mathptmx}
  \usepackage[belowskip=-5pt,aboveskip=5pt]{caption} %

  \usepackage{enumitem}
}

\newcommand{\qed}{\ensuremath{\square}}

\longversion{
\documentclass[10pt,letter]{article}%
  \usepackage[hscale=0.7,scale=0.75]{geometry}
}

\usepackage{url}  %
\usepackage[breaklinks]{hyperref}
\usepackage{multirow}
\usepackage{amssymb,amsfonts,amsmath}
\usepackage{mathtools}

\longversion{
	\usepackage{amsthm}%
}
\usepackage{stackengine}
\usepackage{calc}
\newlength\shlength
\newcommand\xshlongvec[2][0]{\setlength\shlength{#1pt}%
  \stackengine{-5.6pt}{$#2$}{\smash{$\kern\shlength%
    \stackengine{7.55pt}{$\mathchar"017E$}%
      {\rule{\widthof{$#2$}}{.57pt}\kern.4pt}{O}{r}{F}{F}{L}\kern-\shlength$}}%
      {O}{c}{F}{T}{S}}

\usepackage{xspace}
\usepackage[ruled,vlined,linesnumbered]{algorithm2e}
\usepackage{listings}
\lstnewenvironment{aspdrupe}[1][]{%
  \lstset{basicstyle=\ttfamily,columns=fullflexible,keepspaces=false,showstringspaces=true, #1}}
{}

\usepackage[x11names]{xcolor}
\usepackage{graphicx}  %
\renewcommand*{\algorithmcfname}{Listing}
\SetKwInput{KwData}{In}
\SetKwInput{KwResult}{Out}
\SetAlFnt{\small}
\SetAlCapFnt{\small}
\SetAlCapNameFnt{\small}
\SetAlCapHSkip{0pt}
\SetEndCharOfAlgoLine{}
\def\etal{et~al.{}}

\newcommand{\FIX}[1]{{{#1}}}

\makeatletter
\providecommand{\leftsquigarrow}{%
	\mathrel{\mathpalette\reflect@squig\relax}%
}
\newcommand{\reflect@squig}[2]{%
	\reflectbox{$\m@th#1\rightsquigarrow$}%
}
\makeatother

\DeclareRobustCommand\proves{\mathrel{|}\joinrel\mkern-.5mu\mathrel{-}}

\newcommand{\naf}{{\sim}}

\newcommand{\T}{\mathbf{T}}
\newcommand{\F}{\mathbf{F}}
\newcommand{\lcompl}[1]{\overline{#1}}  %
\newcommand{\Res}{\text{Res}}
\newcommand{\UnitPropagation}{\text{UnitPropagation}}
\newcommand{\dl}{\mathit{dl}}
\newcommand{\dlevel}{\mathit{dlevel}}
\newcommand{\atom}{\at}
\newcommand{\NogoodPropagation}{\text{NogoodPropagation}}
\newcommand{\ConflictAnalysis}{\text{ConflictAnalysis}}
\newcommand{\Select}{\text{Select}}
\newcommand{\compdef}{\text{Bdef}}
\newcommand{\UnfoundedSet}{\text{UnfoundedSet}}

\newcommand{\algonew}[1]{{\color{red}#1}}

\newcommand{\dummyvar}{\epsilon}
\newcommand{\proofsteploop}{\ensuremath{\mathsf{l}}}
\newcommand{\proofstepcomp}{\ensuremath{\mathsf{c}}}
\newcommand{\proofstepcompsup}{\ensuremath{\mathsf{s}}}
\newcommand{\proofstepunfound}{\ensuremath{\mathsf{u}}}
\newcommand{\proofstepadd}{\ensuremath{\mathsf{a}}}
\newcommand{\proofstepbody}{\ensuremath{\mathsf{b}}}
\newcommand{\proofstepext}{\ensuremath{\mathsf{e}}}
\newcommand{\proofstepdel}{\ensuremath{\mathsf{d}}}

\renewcommand{\neg}{{\sim}}

\newcommand{\seqconcat}{\ensuremath{\cdot}}  %
\newcommand{\msum}{\ensuremath{\uplus}}
\newcommand{\mdiff}{\ensuremath{\smallsetminus}}

\DeclareMathOperator{\wght}{wght}%
\DeclareMathOperator{\bnd}{bnd}%
\DeclareMathOperator{\optimize}{\leftsquigarrow}

\usepackage{microtype}
\usepackage{tikz}
\usetikzlibrary{shapes}
\usetikzlibrary{calc}

\shortversion{
}

\usetikzlibrary{positioning}

\tikzstyle{tdnode} = [draw,rounded corners,top color=vertexTopColor,bottom color=vertexBottomColor,minimum size=1.5em]
\tikzstyle{stdnode} = [tdnode, font=\scriptsize]
\tikzstyle{stdnodecompact} = [stdnode, inner sep = 1.5pt, outer sep = 0.1pt]
\tikzstyle{stdnodetable} = [stdnode, inner sep = 0.5pt, outer sep = 0]
\tikzstyle{stdnodenum} = [minimum size=1.5em, font=\scriptsize]
\tikzstyle{tdedge} = [-,draw,thick]
\tikzstyle{tdlabel} = [draw=none, rectangle, fill=none, inner sep=0pt, font=\scriptsize]
\colorlet{vertexTopColor}{white}
\colorlet{vertexBottomColor}{black!10}
\usepackage{ifthen}
\newif\iflong
\usepackage{todonotes}

\newcommand{\todoi}[1]{\todo[inline]{#1}}

\newcommand{\future}[1]{}
\newcommand{\friedhof}[1]{}
\def\hy{\hbox{-}\nobreak\hskip0pt}

\newcommand{\ASPRUP}{ASP\hy RUP\xspace}
\newcommand{\ASPDRUPE}{ASP\hy DRUPE\xspace}
\newcommand{\ASPRAT}{ASP\hy RAT\xspace}
\newcommand{\SB}{\{}%
\newcommand{\SM}{\mid}%
\newcommand{\SE}{\}}%

\newcommand{\Card}[1]{\left|#1\right|}

\makeatletter
\newcommand{\algorithmfootnote}[2][\footnotesize]{
  \let\old@algocf@finish\@algocf@finish
  \def\@algocf@finish{\old@algocf@finish
    \leavevmode\rlap{\begin{minipage}{\linewidth}
    #1#2
    \end{minipage}}
  }
}
\makeatother

\newcommand{\por}{\vee}

\newcommand{\eqdef}{\ensuremath{\,\mathrel{\mathop:}=}}
\newcommand{\hsep}{\leftarrow\,}

\newcommand{\at}{\text{\normalfont at}}
\newcommand{\bigO}[1]{\ensuremath{{\mathcal O}(#1)}}

\newcommand{\prog}{\ensuremath{P}}

\newcommand{\eb}{\mathsf{EB}}

\newcommand{\lop}{\mathsf{loop}}

\newcommand{\ebs}{\mathsf{IB}}
\newcommand{\ba}{\mathsf{IAss}}
\newcommand{\vm}{\mathsf{vm}}
\newcommand{\outn}{\mathsf{tv}}

\newcommand{\smof}{\mathsf{smod}}
\newcommand{\ext}{\mathsf{ext}}
\newcommand{\bodies}{\mathtt{Bod}}

\newcommand{\tf}{\T}
\newcommand{\ff}{\F}

\DeclareMathOperator{\progloops}{\lop} %

\newcommand{\Nat}{\mathbb{N}} %

\newtheorem{example}{Example}

\newtheorem{proposition}{Proposition}

\newtheorem{theorem}{Theorem}
\newtheorem{lemma}{Lemma}
\newtheorem{definition}{Definition}
\newtheorem{corollary}{Corollary}
{%
  \addtocounter{observation}{-1}
  \endgroup
}%

{%
  \addtocounter{corollary}{-1}
  \endgroup
}%

{%
  \addtocounter{lemma}{-1}
  \endgroup
}%

{%
  \addtocounter{proposition}{-1}
  \endgroup
}%

{%
  \addtocounter{theorem}{-1}
  \endgroup
}%

\shortversion{
\title[Inconsistency Proofs for ASP: The \ASPDRUPE Format]{%
        Inconsistency Proofs for ASP:\\The \ASPDRUPE Format%
}%

\author[Alviano~\etal]
	{Mario Alviano\\
	University of Calabria, Italy\\
	\email{mario@alviano.net} 
	\and Carmine Dodaro\\
	University of Calabria, Italy\\
	\email{dodaro@mat.unical.it} 
	\and Johannes K. Fichte\\
	TU Dresden, Germany\\
	\email{johannes.fichte@tu-dresden.de} 
	\and Markus Hecher\\
	TU Wien, Austria\\
	\email{hecher@dbai.tuwien.ac.at} 
	\and Tobias Philipp\\
	secunet Security Networks AG, Germany\\\email{tobias.philipp@secunet.com} 
	\and Jakob Rath\\
	TU Wien, Austria\\
	\email{jakob.rath@tuwien.ac.at}}
	
\jdate{March 2003}
\pubyear{2003}
\pagerange{\pageref{firstpage}--\pageref{lastpage}}
\doi{S1471068401001193}
}%
\begin{document}
\maketitle

\begin{abstract}
  Answer Set Programming (ASP) solvers are highly-tuned and complex procedures that implicitly
  solve the consistency problem, i.e., deciding whether a logic program
  admits an answer set.
  Verifying whether a claimed answer set is formally a correct answer
  set of the program can be decided in polynomial time for (normal)
  programs. However, it is far from immediate to verify whether a
  program that is claimed to be inconsistent, indeed does not admit
  any answer sets.  In this paper, we address this problem and develop
  the new proof format \ASPDRUPE for propositional, disjunctive logic programs,
  including weight and choice rules.  \ASPDRUPE is based on the Reverse Unit Propagation (RUP)
  format designed for Boolean satisfiability.  We establish
  correctness of \ASPDRUPE and discuss how to integrate it into modern
  ASP solvers. Later, we provide an implementation of \ASPDRUPE 
  into the wasp solver for normal logic programs.
\end{abstract}


\section{Introduction}

Answer Set Programming (ASP)~\cite{BrewkaEiterTruszczynski11} is a
logic-based declarative modeling language and problem solving
framework~\cite{GebserEtAl12} for hard computational problems and an
active research area in \FIX{artificial intelligence (AI) and knowledge representation and reasoning}. 
It has been applied
both in
academia~\cite{BalducciniGelfondNogueira06a,GebserEtAl10,GebserSchaubThieleVeber11}
and industry~\cite{Gebser2011a,GuziolowskiEtAl13a,RiccaEtAl12}.
In propositional ASP questions are encoded by atoms combined into
rules and constraints which form a logic program. Solutions to the
program consist of sets of atoms called answer sets;
\FIX{if no solutions exist then the program is \emph{inconsistent}}.

\FIX{Knowledge representation languages like ASP are usually considered explainable AI, as they are based on deduction, which is an explainable procedure.
For example, we can easily explain answer sets of a normal logic program in terms of program reducts and fixpoint operators~\cite{LiuEtAl10}.
In this case, we may argue that answer sets are self-explanatory, and therefore ASP systems providing answer sets are explainable AI systems.
However, modern ASP systems do not provide any explanation for inconsistent programs;
there is no witness that can be checked or evidence of the correctness of the refutation of the input program.
Hence, even if inconsistency of logic programs is anyhow explainable with mathematical rigour, ASP systems are essentially black-boxes in this case, and just report the absence of answer sets.
We believe that adding inconsistency proofs in ASP systems is important to make them explainable for inconsistent programs, but also provides auditability for consistent programs. 
Thanks to a duality result in the literature~\cite{Pearce99}, such inconsistency proofs for ASP can be also used as a certificate for the validity of some formulas of intuitionistic logic and other intermediate logics.
A further application of these inconsistency proofs is query answering in ASP, which is usually achieved by inconsistency checks.
There, the goal is to provide a witness for cautiously true answers of a given query.
}

Modern ASP solvers have been highly influenced by SAT solvers, which
solve the Boolean satisfiability problem and are often based on
conflict-driven clause learning~\cite{MarquesSilvaSakallah99}.
Typically, ASP solvers aim for computing an answer set of a given
program, and therefore solve the \emph{consistency problem} that asks whether a given program has an
answer set. This problem is on the second level of the polynomial
hierarchy when allowing arbitrary propositional disjunctive programs
as input and on the first level when restricting to disjunction-free programs~\cite{Truszczynski11}.
As already stated, while consistency of a program can be easily verified given such a
computed answer set, verifying whether a solver correctly outputted
that a program is inconsistent, is not immediate.
Given that ASP solvers are also used for critical
applications~\cite{GebserEtAl18a,HaubeltNeubauerSchaub18a}, their correctness is of utter importance.

When looking at SAT solvers, various techniques have been developed to
ensure correctness of unsatisfiability, such as clausal proof
variants~\cite{Gelder08,GoldbergNovikov03} based on clauses that have
\emph{RUP (reverse unit propagation)} and \emph{RAT (Resolution
  Asymmetric Tautology)} property.
These proof formats share verifiability in polynomial time in the size
of the proof and input formula and can be tightly coupled with modern
solving techniques. A solver outputs such a proof during solving. Thereby, the correctness of solving can be verified by a
relatively simple method for every input instance.
While there are variants of these proofs
for various problems, such as extensions to verify the
validity of quantified Boolean
formulas~\cite{HeuleHuntWetzler13,WetzlerHeuleHunt14}
(QRAT~\cite{HeuleSeidlBiere14} and QRAT+~\cite{LonsingEgly18}), to our
knowledge such a format is not yet available for verifying ASP
solvers. 
\FIX{
One approach to certify inconsistency of a given normal program is to
translate the program into a SAT formula in polynomial
time~\cite{LinZhao03,Janhunen06} and obtain
a proof from a SAT solver, e.g., via a RAT-based proof format,
to verify that indeed the program is inconsistent. 
Unfortunately, this approach does not take the techniques into account
that are employed by state-of-the-art ASP solvers and therefore seem
to lack efficiency and scalability.
Further, this is still not a suitable technique for disjunctive programs,
nor to verify whether
internally the ASP solver is able to \emph{correctly} explain the obtained result.}

We follow this line of research and establish the following
novel results for ASP:
\begin{enumerate}
\item We present the proof format \ASPDRUPE based on RUP for logic programs given in SModels~\cite{lparse} 
\FIX{or ASPIF~\cite{GebserEtAl16} (restricted to ASP without theory reasoning)} input format \FIX{including disjunctive programs}
and show its
  correctness. %
\item We provide an algorithm for verifying that a given solving trace
  in the \ASPDRUPE format is indeed a valid proof for inconsistency of the
  input program. \FIX{This algorithm works in polynomial time in the size of the given solving trace.}
\item We illustrate on an abstract ASP solving algorithm how one can
  integrate \ASPDRUPE into state-of-the-art ASP solvers like
  clasp~\cite{GebserEtAl12} and wasp~\cite{AlvianoEtAl15}.
\item \FIX{We provide an implementation in a variant of wasp, where \ASPDRUPE is integrated for normal ASP. 
This variant of wasp is able to not only explain inconsistency for inconsistent logic programs, but also provides auditability in case of consistency for verifying whether the provided answer set was indeed correctly obtained.}
\end{enumerate}

\paragraph{Related Work.}
Heule~\etal~\cite{HeuleHuntWetzler13} presented a proof format based
on the RAT property and subsequently a program to validate solving
traces in this format~\cite{WetzlerHeuleHunt14}.
Extended resolution allows to polynomially simulate the DRAT
format~\cite{KieslEtAl18} and vice-versa~\cite{WetzlerHeuleHunt14}.
Many advanced techniques, such as \emph{XOR
  reasoning}~\cite{PhilippRebola16} as well as
\emph{symmetry breaking}~\cite{HeuleHuntWetzler15} can be
expressed in DRAT and efficient, verified checkers based on RAT have
been developed~\cite{CruzFilipeEtAl17}.
Further, RAT is also available for QBF~\cite{HeuleSeidlBiere14} and
has been extended to cover a more powerful redundancy
property~\cite{LonsingEgly18}.

\section{Preliminaries}
\subsection{Answer Set Programming (ASP)}
We follow standard definitions of propositional ASP~\cite{BrewkaEiterTruszczynski11\longversion{,JanhunenNiemela16a}}
and use rules defined by the SModels~\cite{lparse} \FIX{or ASPIF~\cite{GebserEtAl16} (restricted to ASP without theory reasoning)} input
format, which is widely supported by modern ASP solvers.
In particular, let $\ell$, $m$, $n$ be non-negative integers such that
$1 \leq \ell \leq m \leq n$, $a_1$, $\ldots$, $a_n$ %
propositional
atoms, %
$w$, $w_1$, $\ldots$, $w_n$ non-negative integers.
A \emph{choice rule} is an expression of the form
$\{a_1; \ldots; a_\ell \} \hsep a_{\ell+1}, \ldots, a_m, \neg a_{m+1},
\ldots, \neg a_n$,
a \emph{disjunctive rule} is of the form
$a_1\por \cdots \por a_\ell \hsep a_{\ell+1}, \ldots, a_{m}, \neg
a_{m+1},$ $\ldots, \neg a_n$, and
a \emph{weight rule} is of the form
$a_{\ell} \hsep w \leqslant \{ a_{\ell + 1} = w_{\ell + 1}, \ldots,
a_m = w_m,%
\, \neg a_{m+1} = w_{m+1}, \ldots, \neg a_n = w_n \}$, where~$\ell=1$.
A \emph{rule} is either a disjunctive, a choice, or a weight rule. %
A \emph{\FIX{(disjunctive logic)} program}~$\prog$ is a {finite} set of rules.
For a rule~$r$, we let $H_r \eqdef \{a_1, \ldots, a_\ell\}$,
$B^+_r \eqdef \{a_{\ell+1}, \ldots, a_{m}\}$,
$B^-_r \eqdef \{a_{m+1}, \ldots, a_n\}$, and
$B_r \eqdef B^+_r \cup \{\neg a \mid a \in B^-_r\}$ is a set of
\emph{literals}, i.e., an atom or the negation thereof.
We denote the sets of \emph{atoms} occurring in a rule~$r$ or in a
program~$\prog$ by $\at(r) \eqdef H_r \cup B^+_r \cup B^-_r$ and
$\at(\prog)\eqdef \bigcup_{r\in\prog} \at(r)$.
\FIX{For a weight rule~$r$, let $\wght(r,l)$ map literal~$l$ %
to its weight~$w_i$ in rule~$r$ if $l=a_i$ for
$\ell+1 \leq i \leq m$, or if~$l=\neg a_i$ for~$m+1\leq i\leq n$, 
and to $0$ otherwise. Further, let
$\wght(r,L) \eqdef \Sigma_{l \in L} \wght(r,l)$ for a set~$L$ of literals
and let $\bnd(r)\eqdef w$ be its \emph{bound}.}
A \emph{normal} (logic) program~$\prog$
 is a disjunctive program~$\prog$
 with $\Card{H_r} \leq 1$ for every~$r \in \prog$.
 The \emph{positive dependency digraph}~$D_\prog$ of
 $\prog$ is the digraph defined on the set~$\SB a \SM a \in H_r \cup
 B^+_r, r \in P \}$ of atoms, where for every rule~$r \in
 \prog$ two atoms $a\in B^+_r$ and~$b\in
 H_r$ are joined by an
 edge~$(a,b)$.  We denote the set of all cycles (loops)
 in~$D_\prog$ by~$\lop(\prog)$. A
 program~$\prog$ is called~\emph{tight},
 if~$\lop(\prog) = \emptyset$. 
 \FIX{While we allow programs with loops that might also involve atoms of weight rules,
 we consider weight rules only as a compact representation of a set of normal rules, similar to the definition of stable models in related work~\cite{BomansonEtAl16}.
 In other words, we do not consider advanced semantics concerning recursive weight rules (recursive aggregates).
 In case of solvers with different semantics, one can restrict the input language to disregard recursive weight rules, 
 which is also in accordance with the latest ASP-Core-2.03c standard~\cite{aspcore2}.
 This restriction is motivated by a lack of consensus on the interpretation 
 of recursive weight rules~\cite{Ferraris11,FaberPfeiferLeone11,GelfondZhang14,PelovDeneckerBruynooghe07,SonPontelli07}.
}
\subsection{Solving Logic Programs}
Let~$P$ be a given program, $r\in\prog$ be a rule, and~$a\in\at(\prog)$. %
We define the \emph{set~$\ebs(r,a)$ of induced
  bodies} with~$a$ in the head %
by the singleton $\{B_r\mid a\in H_r\}$ if
$r$ is a choice rule, by~$\{B_r \cup \{\neg b \mid b \in H_r\setminus\{a\}\} \mid a\in H_r\}$ if~$r$ is a disjunctive rule, and by the union over~$\SB \{A\cap
B^+_r\SE \cup \{\neg b\SM b\in A\cap
B_r^-\}\mid a\in H_r\SE$ for every (subset-minimal)
set~$L$ of literals such that $\wght(r,L)\geq \bnd(r)$,
if~$r$ is a weight rule. %
This allows us to define~$\ebs(\prog,a)\eqdef \bigcup_{r\in\prog,a\in H_r}\ebs(r,a)$, and $\bodies(\prog)\eqdef
\bigcup_{a\in \at(\prog)}\ebs(\prog,a)$.
A \emph{variable assignment} is
either~$\tf X$ or~$\ff X$, where
\emph{variable}~$X$ is either an atom, or an induced body, or a
\emph{fresh atom} that does not occur
in~$\at(\prog)$.  For a variable assignment~$l$,
$\lcompl{l}$ is the \emph{complementary variable assignment} of
$l$, i.e., $\lcompl{l} \eqdef \F X$ if $l=\T X$ and $\lcompl{l} \eqdef
\T X$ if $l=\F X$. %
An \emph{assignment}~$A$ is a set of variable assignments, where~$A^\T\eqdef \SB X \SM \T X \in A\SE$, $A^\F\eqdef \SB X \SM \F X \in A\SE$, and $\lcompl{A}\eqdef\{\lcompl{l}\mid l\in A\}$ such that~$A^\T\cap A^\F=\emptyset$. %
For a set~$B$ of literals, we define the \emph{induced assignment}~$\ba(\prog, B)\eqdef \SB \T a \SM a \in \at(\prog) \cap B\SE \cup \SB \F a \SM \neg a\in B\SE$.
A \emph{nogood}~$\delta$ is an assignment, %
which is not allowed, where~$\Box\eqdef \emptyset$ refers to the empty nogood. %
Given a set~$\Delta$ of nogoods. 
We define the least fixpoint~$\Delta_{\proves_1}$ of \emph{unit propagated nogoods} by the fixpoint computation~$\Delta_{\proves_1}^0\eqdef \Delta$, and~$\Delta_{\proves_1}^i\eqdef \Delta_{\proves_1}^{i-1} \cup \{\delta \setminus\{l\} \mid \delta\in\Delta_{\proves_1}^{i-1}, l\in\delta, \{\lcompl{l}\} \in \Delta_{\proves_1}^{i-1}\}$ for~$i\geq 1$.
Nogood~$\delta$ is a \emph{consequence using unit propagation (UP)} of set~$\Delta$, denoted by~$\Delta\proves_1\delta$, if~$\delta\in\Delta_{\proves_1}$.
An assignment~$A$ \emph{satisfies} a set~$\Delta$ of nogoods (written~$A\models \Delta$)
if for every~$\delta\in \Delta$, we have~${\delta}\not\subseteq A$.
Set~$\Gamma$ of nogoods is a \emph{consequence} of a set~$\Delta$ of nogoods (denoted by $\Delta \models \Gamma$) if every assignment, which contains a variable assignment for all variables in~$\Delta\cup\Gamma$, that satisfies~$\Delta$ also satisfies~$\Gamma$.  
The
set~$\Delta_\prog$ of \emph{completion
  nogoods~\cite{Clark77,GebserEtAl12}}
is %
defined by~$\Delta_\prog\eqdef\Delta_\prog^{\compdef} \cup \Delta_\prog^\rightarrow \cup \Delta_\prog^\leftarrow$, where
\begin{align*}
  \Delta_\prog^{\compdef} \eqdef& \bigcup_{B\in\bodies(\prog)} \Big\SB \{\ff B\} \cup \ba(\prog,B)\Big\SE \cup \Big\SB\{\tf B, \lcompl{l}\} \mid l \in {\ba(\prog,B)}\Big\SE \; \\ %
  \Delta_\prog^\rightarrow\eqdef
                     &\bigcup_{a\in\at(\prog)} \Big\SB\{\tf a\} \cup \{\ff B \SM B \in \ebs(P,a) \}\Big\SE,
		     \Delta_\prog^\leftarrow\eqdef
                       \hspace{-2em}\bigcup_{\substack{\text{non-choice rule } r\in \prog,\\ a\in H_r%
                       }}  \Big\SB \{\ff a, \T B \} \SM B \in \ebs(r,a) \Big\SE.
\end{align*}
\FIX{Note that in practice, current ASP solvers do not fully compute~$\Delta_\prog$.
Instead, these solvers partially compute $\Delta_\prog$ and %
add relevant nogoods lazily during solving~\cite{AlvianoDodaroMaratea18}.}

Then, if $\prog$ is tight the set~$A^\T \cap \at(\prog)$ is an
\emph{answer-set} if and only if there is a satisfying assignment~$A$
for~$\Delta_\prog$~\cite{Fages94,GebserEtAl12}.
The set~$\eb(\prog,U)$ of \emph{external bodies} of
program~$\prog$ and
set~$U\subseteq\at(\prog)$ of atoms are given by~$\eb(\prog,U)\eqdef
\SB B \SM B\in\ebs(\prog,a), B \cap U = \emptyset, a\in U\SE$~\cite{GebserEtAl12}.
We define the \emph{loop nogood}~$\lambda(a,U)$ for an atom~$a \in
U$ on a loop~$U \in \lop(\prog)$ by $\lambda(a,U) \eqdef \SB \tf a \SE
\cup \SB \ff B \SM B \in \eb(\prog,U) \SE$.
For a logic program $P$, the set~$A^\T \cap
\at(\prog)$ is an \emph{answer set} if and only if there is a
satisfying assignment~$A$ for~$\Delta_\prog\cup
\Lambda_P$, where $\Lambda_P\eqdef \SB \lambda(a,U) \SM
a\in U, U\in\lop(\prog)\SE$~\cite{LinZhao03,Faber05}.

\section{\ASPDRUPE: RUP-like Format for Proof Logging}
Inspired by RUP-style unsatisfiability proofs in the field of Boolean
satisfiability solving~\cite{GoldbergNovikov03}, we aim for a proof of
inconsistency of a program. 
Since modern ASP solvers use \emph{Clark's completion}~\cite{Clark77} to
transform a program into a set of nogoods, we do so as well. 
Our aimed proof then has the following features:
\begin{enumerate}
\item \emph{Existence of a simple verification algorithm.} In order to
  increase confidence in the correctness of results, the algorithm
  that verifies the proof has to be fairly easy to understand and to
  implement.
\item \emph{Low complexity.} The proof is verifiable in polynomial
  time in its length and the size of the completion nogoods.
\item \emph{Integrability into solvers that employ Conflict-Driven
    Nogood Learning (CDNL).} The proof can stepwise be outputted
  during solving with minimal impact on the solving algorithm and
  hence the solver.
\end{enumerate}

\noindent
The method works as follows: We run the solver on the set~$\Delta_\prog$ of completion nogoods for given input
program~$P$. The solver outputs either an answer set or that $P$ has
no answer set and a proof~$\Pi$.  We pass $P$ together with $\Pi$ to
the verifier in order to validate whether the solver's assessment is
in accordance with its~outputted~proof.

\subsection{The Proof Format for Logic Programs}
The basic idea of clausal proofs for SAT is the following: One starts
with the input formula in CNF (given as a set of clauses). Every step
of the proof denotes an addition or deletion of a clause to/from the
set of clauses.  For additions, the condition is to only add clauses
that are a logical consequence of the current set of clauses and that it
can be checked easily,~e.g., use only unit propagation.

For our format \ASPDRUPE, we consider Clark's completion $\Delta_P$ as
the initial set of nogoods corresponding to the input program $P$.
Besides addition and deletion of nogoods, we need proof steps that
model how the solver excludes unfounded sets (loops).

\newcommand{\progex}{P}
\newcommand{\rsep}{;\;}
\begin{example}
    \label{ex:program}
    Consider program $\progex = \SB %
    a \leftarrow b, d \rsep %
    b \leftarrow a, d \rsep %
    a \leftarrow c\rsep %
    b \leftarrow c, d \rsep %
    c \leftarrow \naf d \rsep %
    d \leftarrow \naf c \rsep %
    e \leftarrow c, \naf e \rsep %
    e \leftarrow \naf a, \naf e \SE$, which is inconsistent. 
    $P$ contains only the positive loop $L = \{a,b\}$, whose external
    support is given by the
    set~$\SB a \leftarrow c\rsep b \leftarrow c,d \SE$ of rules, and
    thus $\eb(\progex, L) = \{ \{c\}, \{c,d\} \}$.  Set~$L$ induces
    two possible loop nogoods,
    $\lambda(a,L) = \{ \tf a, \ff \{c\}, \ff \{c,d\} \}$ and
    $\lambda(b,L) = \{ \tf b, \ff \{c\}, \ff \{c,d\} \}$.
\end{example}

We describe the
proof format \ASPDRUPE for logic programs and %
adapt the RUP property~\cite{GoldbergNovikov03} to nogoods as follows.

\begin{definition}[nogood RUP]
  Let $\Delta$ be a set of nogoods. Then, a nogood~$\delta$ is
  \emph{RUP} (reverse unit propagable) \emph{for~$\Delta$} if~$\Delta \cup \{\{\lcompl{l}\}\mid l\in\delta\} \proves_1 \Box$, i.e., we can derive~$\Box$ using only unit propagation.
\end{definition}
A \emph{proof step} is a triple $(t, \delta, a)$, where
$t \in \{ \proofstepadd, \proofstepcomp, \proofstepcompsup, \proofstepext, \proofstepdel, \proofsteploop \}$
denotes the \emph{type} of the step, $\delta$ is an assignment,
and $a$ is an atom or~$\dummyvar$. %
The
type~$t \in \{ \proofstepadd, \proofstepcomp, \proofstepcompsup, \proofstepext, \proofstepdel, \proofsteploop \}$
indicates whether the step is an
addition~(\proofstepadd), 
a completion rule addition~(\proofstepcomp),
a completion support addition~(\proofstepcompsup),
an extension~(\proofstepext), a
deletion~(\proofstepdel), or a loop addition~(\proofsteploop). %
A \emph{proof sequence} for a logic program $P$ is a finite sequence
$\Pi\eqdef\langle \sigma_1, \ldots, \sigma_n \rangle$ of proof steps.
Initially, a proof sequence gets associated with a
set~$\nabla_0(\Pi) \eqdef \Delta_P^{\compdef}$ of nogoods. 
Note that although the set $\Delta_P^{\compdef}$ might be exponential
in the size of the program~$\prog$, body definitions
for %
body variables~$B\in\bodies(\prog)$ that do not occur in the proof are never materialized.
Then, each proof
step~$\sigma_i$ for~$1\leq i\leq n$ subsequently
transforms~$\nabla_{i-1}(\Pi)$ into the induced set~$\nabla_i(\Pi)$ of nogoods,
formally defined below.
An \emph{\ASPDRUPE derivation} is a proof sequence that allows for
RUP-like rules for ASP and includes both deletion and extension. In an
\ASPDRUPE derivation each step~$\sigma=(t, \delta, a)$ has to satisfy
a condition depending on its type as follows:
\begin{enumerate}
\item An \emph{addition}~$\sigma = (\proofstepadd, \delta, \dummyvar)$ inserts
  a nogood~$\delta$ that is RUP for~$\nabla_{i-1}(\Pi)$.
\item A \emph{completion rule addition}~$\sigma=(\proofstepcomp, \delta, \dummyvar)$ inserts a nogood~$\delta\in\Delta_\prog^\leftarrow$.
\item A \emph{completion support addition}~$\sigma=(\proofstepcompsup, \{\F B_1, \ldots, \F B_k\}, a)$ inserts a nogood~$\{\T a, \F B_1,\ldots,$ $\F B_k\}\in\Delta_\prog^{\rightarrow}$ if~$\{B_1,\ldots,B_k\}=\ebs(\prog,a)$.
\item An \emph{extension}~$\sigma = (\proofstepext, \delta, a)$
  introduces a definition that renders nogood~$\delta$ equivalent to a
  fresh atom~$a$, i.e., $a$ does not appear in~$\bigcup_{j=0}^{j=i-1}\nabla_{j}(\Pi)\cup\Delta_\prog$.
  Formally, this rule represents the set
  $\ext(a,\delta)\eqdef \big\SB \delta \cup \{\F a\}\big\SE \cup
  \big\SB\{\T a, \lcompl{l}\} \SM l \in \delta \big\SE$ of
  \emph{extension nogoods}.
\item A \emph{deletion}~$\sigma = (\proofstepdel, \delta, \dummyvar)$
  represents the deletion of $\delta$ from~$\nabla_{i-1}(\Pi)$.
\item A \emph{loop
    addition\footnote{%
    There could be an exponential number of external bodies involving weight rules. 
    However, both clasp and wasp treat weight rules differently~\cite{AlvianoEtAl15}. 
    Alternatively, one could easily modify the loop addition type to list also involved external bodies (as in the completion support addition type), which we did not for the sake of readability.%
    }}~$\sigma=(\proofsteploop, \{ \T a_1, \ldots, \T a_k \},
  a_1)$ inserts a loop nogood $\lambda(a_1, L)$ for a loop
  $L= \{ a_1, \ldots, a_k \} \in \lop(P)$. %
\end{enumerate}
Given an \ASPDRUPE
derivation~$\Pi = \langle \sigma_1, \ldots,\sigma_n\rangle$, we define
the \emph{set~$\nabla_i(\Pi)$ of nogoods induced by step~$i$} as the
result of applying proof
steps~$\langle\sigma_1, \ldots, \sigma_i\rangle$ to the initial
completion body definitions~$\Delta_P^{\compdef}$ for~$0\leq i\leq n$. For our inductive
definition in the following, we use multiset semantics for additions
and deletions of nogoods, and write $\msum$ for the multiset sum.
\begin{align*}
  \nabla_0(\Pi) &\coloneqq \Delta_P^{\compdef} \\
  \nabla_{i}(\Pi) &\coloneqq
               \begin{cases}
                 \nabla_{i-1}(\Pi) \msum \{ \delta \}, & \text{if
                   $\sigma_i = (\proofstepadd, \delta, \dummyvar)$,}
                 \\
		\nabla_{i-1}(\Pi) \msum \{ \delta \}, & \text{if
                   $\sigma_i = (\proofstepcomp, \delta, \dummyvar)$,}
                 \\
		\nabla_{i-1}(\Pi) \msum \{ \delta\cup\{\T a\}\}, & \text{if
                   $\sigma_i = (\proofstepcompsup, \delta, a)$,}
                 \\
                 \nabla_{i-1}(\Pi) \msum \ext(a,\delta), & \text{if
                   $\sigma_i = (\proofstepext, \delta, a)$,}
                 \\
                 \nabla_{i-1}(\Pi) \mdiff \{ \delta \}, & \text{if
                   $\sigma_i = (\proofstepdel, \delta, \dummyvar)$,}
                 \\
                 \nabla_{i-1}(\Pi) \msum \{ \lambda(a_1, \{a_1, \ldots,
                 a_k\}) \}, & \text{if
                   $\sigma_i = (\proofsteploop, \{\T a_1, \ldots, \T
                   a_k\}, a_1)$.}
    \end{cases}
\end{align*}
Then, we say that an \ASPDRUPE derivation~$\Pi$ is an \emph{\ASPDRUPE proof}
for the inconsistency of~$P$ if it actually derives inconsistency
for~$\prog$, formally, $\Box \in \nabla_n(\Pi)$.
\FIX{Note that $\Delta_P^{\compdef}$ might be exponential in the input program size,
in the worst case. However, there is no need to materialize the 
set~$\Delta_P^{\compdef}$, %
as, intuitively, this set of body definitions %
only ensures that every induced body has a reserved
auxiliary atom that can be used to ``address''
the body in a compact way.
In an actual implementation of a solver that uses \ASPDRUPE, %
one needs to specify these used auxiliary atoms anyway, cf. Section~\ref{sec:implspec},
where implementational specifications of \ASPDRUPE are described.
}

\begin{example}
    \label{ex:proof}
    Consider program~$\progex$ from Example~\ref{ex:program} and
    loop~$L=\{a,b\}$, which induces loop nogood
    $\lambda(a,L) = \{ \tf a, \ff \{c\}, \ff \{c,d\} \}$.
    Then, the proof sequence~$\Pi=\langle \sigma_1,\ldots,\sigma_{18}\rangle$ is an \ASPDRUPE proof for the inconsistency of $\progex$ with
    \begin{align*}
        \sigma_1 & = \FIX{ ( \proofstepcompsup, \{ \F \{b,d\}, \F \{c\} \}, a ) }
                 & \nabla_1(\Pi) & = \Delta_{\progex}^{\compdef} \msum \{ \{ \T a, \F \{b,d\}, \F \{c\} \} \}
                 \\
        \sigma_2 & = ( \proofstepcomp, \{ \F c, \T \{\naf d\} \}, \dummyvar )
                 & \nabla_2(\Pi) & = \nabla_1(\Pi) \msum \{ \{ \F c, \T \{\naf d\} \} \}
                 \\
        \sigma_3 & = \FIX{ ( \proofstepcomp, \{ \F d, \T \{\naf c\} \}, \dummyvar ) }
                 & \nabla_3(\Pi) & = \nabla_2(\Pi) \msum \{ \{ \F d, \T \{\naf c\} \} \}
                 \\
        \sigma_4 & = ( \proofstepcomp, \{ \F b, \T \{a,d\} \}, \dummyvar )
                 & \nabla_4(\Pi) & = \nabla_3(\Pi) \msum \{ \{ \F b, \T \{a,d\} \} \}
                 \\
        \sigma_5 & = ( \proofstepcompsup, \{\F\{c,\naf e\}, \F\{\naf a, \naf e\}\}, e )
                 & \nabla_5(\Pi) & = \nabla_4(\Pi) \msum \{ \{ \T e, \F\{c,\naf e\}, \F\{\naf a, \naf e\} \} \}
                 \\
        \sigma_6 & = ( \proofsteploop, \{ \T a, \T b\}, a )
                 & \nabla_6(\Pi) & = \nabla_5(\Pi) \msum \{ \{ \T a, \F \{c\}, \F \{c,d\} \} \}
                 \\
        \sigma_7 & = ( \proofstepadd, \{ \F \{c\}, \T a \}, \dummyvar )
                 & \nabla_7(\Pi) & = \nabla_6(\Pi) \msum \{ \{ \F \{c\}, \T a \} \}
                 \\
        \sigma_8 & = \FIX{ ( \proofstepcompsup, \{\F\{\naf d\}\}, c ) }
                 & \nabla_8(\Pi) & = \nabla_7(\Pi) \msum \{ \{ \T c, \F\{\naf d\} \} \}
                 \\
        \sigma_9 & = ( \proofstepcompsup, \{\F\{\naf c\}\}, d )
                 & \nabla_9(\Pi) & = \nabla_8(\Pi) \msum \{ \{ \T d, \F\{\naf c\} \} \}
                 \\
        \sigma_{10} & = ( \proofstepcompsup, \{\F\{a,d\}, \F\{c,d\} \}, b )
                    & \nabla_{10}(\Pi) & = \nabla_9(\Pi) \msum \{ \{ \T b, \F\{a,d\}, \F\{c,d\} \} \}
                    \\
        \sigma_{11} & = ( \proofstepadd, \{ \T e, \F \{\naf a, \naf e\} \}, \dummyvar )
                    & \nabla_{11}(\Pi) & = \nabla_{10}(\Pi) \msum \{ \{ \T e, \F \{\naf a, \naf e\} \} \}
                    \\
        \sigma_{12} & = ( \proofstepcomp, \{ \F e, \T \{c,\naf e\} \}, \dummyvar )
                    & \nabla_{12}(\Pi) & = \nabla_{11}(\Pi) \msum \{ \{ \F e, \T \{c,\naf e\} \} \}
                    \\
        \sigma_{13} & = ( \proofstepadd, \{ \T a \}, \dummyvar )
                    & \nabla_{13}(\Pi) & = \nabla_{12}(\Pi) \msum \{ \{ \T a \} \}
                    \\
        \sigma_{14} & = ( \proofstepcomp, \{ \F a, \T \{b,d\} \}, \dummyvar )
                    & \nabla_{14}(\Pi) & = \nabla_{13}(\Pi) \msum \{ \{ \F a, \T \{b,d\} \} \}
                    \\
        \sigma_{15} & = ( \proofstepcomp, \{ \F a, \T \{c\} \}, \dummyvar )
                    & \nabla_{15}(\Pi) & = \nabla_{14}(\Pi) \msum \{ \{ \F a, \T \{c\} \} \}
                    \\
        \sigma_{16} & = ( \proofstepadd, \{ \T e \}, \dummyvar )
                    & \nabla_{16}(\Pi) & = \nabla_{15}(\Pi) \msum \{ \{ \T e \} \}
                    \\
        \sigma_{17} & = ( \proofstepcomp, \{ \F e, \T \{\naf a,\naf e\} \}, \dummyvar )
                    & \nabla_{17}(\Pi) & = \nabla_{16}(\Pi) \msum \{ \{ \F e, \T \{\naf a,\naf e\} \} \}
                    \\
        \sigma_{18} & = ( \proofstepadd, \Box, \dummyvar )
                    & \nabla_{18}(\Pi) & = \nabla_{17}(\Pi) \msum \{ \Box \}.
    \end{align*}
    We show that the proof step $\sigma_7$ is correct, i.e.,
    that $\{ \F \{c\}, \T a \}$ is RUP for $\nabla_6(\Pi)$.
    To this end, we need to derive $\Box$ from
    $\nabla_6(\Pi) \cup \{  \{ \T \{c\} \},  \{ \F a \}  \}$
    by unit propagation.
    With the nogood $\{ \F\{c\}, \T c \} \in \Delta_\progex^\compdef$,
    we derive the unit nogood $\{ \T c \}$.
    With $\{ \T\{c,d\}, \F c \} \in \Delta_\progex^\compdef$
    we now get $\{ \T\{c,d\} \}$.
    With these unit nogoods, $\lambda(a,L)$ reduces to $\Box$.
\end{example}

\subsection{Correctness of \ASPDRUPE}

Next, we establish soundness and completeness of the~\ASPDRUPE format. 

\newcommand{\InvHelper}{\Delta_\prog \cup \Lambda_\prog \cup D_n}%

\begin{lemma}[Invariants]\label{lem:invariant_rup_del}
  Let $\prog$ be a logic program and %
  $\Pi = \langle \sigma_1, \ldots, \sigma_n \rangle$ be a finite \ASPDRUPE
  derivation for program~$P$.  Moreover, let $D_{i}$ be the
  accumulated set of nogoods introduced by the extension rules in
  $\sigma_k$ for all $k \in \{1, \ldots, i\}$. Then, the following
  holds:
  $\Delta_\prog \cup \Lambda_\prog \cup D_i \models \nabla_i(\Pi)
  \text{ for all } i \in \{0, \ldots, n\}$.
\end{lemma}

\begin{proof}[Proof (Sketch).]\renewcommand{\proofbox}{}
  We proceed by induction over the length~$n$ of the derivation.  For
  the base case, we have $\nabla_0(\Pi) = \Delta_P^\compdef$. Hence,
  $\nabla_0(\Pi) \subseteq \Delta_\prog \cup \Lambda_\prog \cup
  D_0$ and the claim holds trivially.
  For the induction step, we assume that the statement holds for
  length~$i$ and consider step $\sigma_{i + 1}$. It remains to do a
  case distinction for the type:
  \begin{enumerate}
  \item Deletion with $\sigma_{i+1}=(\proofstepdel,\delta,\dummyvar)$:
    Immediately, we have $\nabla_i(\Pi) \models \nabla_{i+1}(\Pi)$. Thus,
    transitivity of~$\models$ and the induction hypothesis
    establishes this case.
  \item Addition with $\sigma_{i+1}=(\proofstepadd,\delta,\dummyvar)$:
    Since~$\delta$ is RUP for~$\nabla_i(\Pi)$, we know that $\{\delta\}$ is
    a logical consequence of $\Delta_i$.  The remaining steps to draw
    the conclusion are similar to the deletion step~case.
  \item Completion Rule Addition with $\sigma_{i+1}=(\proofstepcomp,\delta,\dummyvar)$:
    Since the resulting nogood~$\delta\in\Delta_P^\leftarrow$ is contained in~$\Delta_P$, 
    the result follows.
\item Completion Support Addition with $\sigma_{i+1}=(\proofstepcompsup,\delta,a)$:
    Since the resulting nogood~$\{\T a\}\cup\delta$ %
    is contained in~$\Delta_P^\rightarrow$, the result follows.
  \item Extension with $\sigma_{i+1}= (\proofstepext,\delta,a)$:
    According to the induction hypothesis we have
    $\Delta_\prog \cup \Lambda_\prog \cup D_i \models \nabla_i(\Pi)$.
    \FIX{Then, $D_i \subset D_{i+1}$, since $a$ is a fresh variable and $\ext(a,\delta)\in D_{i+1}\setminus D_i$}.
    As $\models$ is monotone, and $D_i \subset D_{i+1}$, we know that
    $\Delta_\prog \cup \Lambda_\prog \cup D_{i+1} \models \nabla_i(\Pi)$.
    It then follows that
    $\Delta_\prog \cup \Lambda_\prog \cup D_{i+1} \models
    \nabla_{i+1}(\Pi)$.
  \item Loop addition with $\sigma_{i+1}=(\proofsteploop,L,a_1)$: By
    definition nogood~$\delta\eqdef\lambda(a_1,L)$ is already
    contained in $\Delta_\prog \cup \Lambda_\prog \cup D_i$, which
    immediately establishes this case. \quad\qed%
\end{enumerate}%
\end{proof}

\begin{theorem}[Soundness and Completeness]\label{ref:normalasp}
  Let $P$ be a logic program. Then, $P$ is inconsistent if and only
  if there is an \ASPDRUPE~proof~for~$P$.
\end{theorem}
\begin{proof}[Proof (Sketch).]
  Let $P$ be a logic program. %
  ``$\Leftarrow$'':
  Assume there is an \ASPDRUPE proof of $P$.  By definition, there is
  a finite sequence of proof steps
  $\langle \sigma_i, \ldots, \sigma_n \rangle$ such that
  $\Box \in \nabla_n(\Pi)$ and $\nabla_n(\Pi)$ is inconsistent.  From
  Lemma~\ref{lem:invariant_rup_del}, we obtain that $\InvHelper$ is
  inconsistent.  As $D$ consists of extension nogoods with disjoint
  variables, we know that $\Delta_\prog \cup \Lambda_\prog$ is
  inconsistent.  We conclude from an earlier
  result~\cite[Theorem~5.4]{GebserEtAl12} that~$P$ is inconsistent.
  ``$\Rightarrow$'': Suppose $P$ is inconsistent.  According to
  earlier work~\cite[Theorem~5.4]{GebserEtAl12}, we know that
  $\Delta_\prog \cup \Lambda_\prog$ is inconsistent.  RUP is
  complete~\cite{Gelder08,GoldbergNovikov03}, which means that for
  every propositional, unsatisfiable formula~$F$ there is a RUP proof
  for~$F$. Hence, we can construct an \ASPDRUPE proof for~$\prog$ as
  follows: 
  (i) Output all completion rule additions for~$\Delta_\prog^\leftarrow$ and completion support additions for~$\Delta_\prog^\rightarrow$.
  (ii) Generate loop addition steps for all
  loops~$L\in\lop(\prog)$.
  (iii) Transform~$\Delta_\prog \cup \Lambda_\prog$ into a
  propositional formula~$F = \lcompl{\Delta_\prog \cup \Lambda_\prog}$
  by inverting all nogoods. 
  (iv) Construct and use a RUP proof for~$F$. Then, output addition
  rules accordingly, where again all clauses need to be inverted to
  obtain addition proof steps using nogoods. 
\end{proof}

Note that in the only-if direction of the proof, one can also use
RAT~\cite{WetzlerHeuleHunt14} proofs without deletion information and
afterwards translate RAT steps into extended resolution
steps~\cite{KieslEtAl18}.

Listing~\ref{algo:proof-checker} presents the \ASPDRUPE checker, that decides whether a given \ASPDRUPE proof is correct.
The input to the checker is both the original program $P$ and the proof $\Pi$.
To check the proof, we encode $P$ into nogoods $\Delta_P$ and then check each statement $\sigma \in \Pi$ sequentially.

\begin{algorithm}[t]
    \caption{\ASPDRUPE-Checker}
    \label{algo:proof-checker}
    \DontPrintSemicolon
    \SetKwInOut{Input}{Input}
    \SetKwInOut{Output}{Output}
    \SetKwBlock{Loop}{loop}{end}
    \SetKwFor{Foreach}{for each}{do}{end}

    \Input{A logic program $P$ and an \ASPDRUPE derivation $\Pi = \langle \sigma_1, \ldots, \sigma_n \rangle$.}
    \Output{\emph{Success} if $\Pi$ proves that $P$ has no answer set, \emph{Error} otherwise.}

    $\nabla \coloneqq \Delta_P^\compdef$\;
    \For{$i = 1, \ldots, n$}{%
         \uIf{$\sigma_i = (\proofstepadd, \delta, \dummyvar)$ \textbf{\textup{and}} $\nabla \proves_1 \delta$}{%
            $\nabla \coloneqq \nabla \msum \{ \delta \}$%
        }
	\uElseIf{$\sigma_i = (\proofstepcomp, \delta, \dummyvar)$ \textbf{\textup{and}} $\delta\in\Delta_\prog^\leftarrow$}{%
            $\nabla \coloneqq \nabla \msum \{ \delta \}$%
        }
        \uElseIf{$\sigma_i = (\proofstepcompsup, \delta, a)$ \textbf{\textup{and}} $\delta=\{ \F B \mid B\in\ebs(\prog,a)\}$}{%
            $\nabla \coloneqq \nabla \msum \{ \delta\cup\{\T a\} \}$%
        }
        \uElseIf{$\sigma_i = (\proofstepext, \delta, a)$ \textbf{\textup{and}} $a$ \textup{is a fresh atom w.r.t. $\bigcup_{j=0}^{j=i-1}\nabla_j(\Pi)\cup\Delta_\prog$}}{%
            $\nabla \coloneqq \nabla \msum \ext(a,\delta)$%
        }
        \uElseIf{$\sigma_i = (\proofstepdel, \delta, \dummyvar)$}{%
            $\nabla \coloneqq \nabla \mdiff \{ \delta \}$%
        }
        \uElseIf{$\sigma_i = (\proofsteploop, \{ \T a_1, \ldots, \T a_k \}, a_1)$}{%
            $U \coloneqq \{a_1, \ldots, a_k\}$\;%
            \leIf{$U \in \progloops(P)$}{%
                $\nabla \coloneqq \nabla \msum \{ \lambda(a_1, U) \}$%
            }{%
                \Return $\mathit{Error}$%
            }
        }
        \Else{%
            \Return $\mathit{Error}$\;
        }
    }

    \leIf{$\Box \in \nabla$}{%
        \Return $\mathit{Success}$%
    }{%
        \Return $\mathit{Error}$%
    }
\end{algorithm}  %

\begin{lemma}
  For a given logic program~$\prog$ and an~\ASPDRUPE
  derivation~$\Pi$, the~\ASPDRUPE-Checker runs in time at
  most~$\Card{\Pi}^{\bigO{1}}$.
\end{lemma}

\begin{corollary}
  Given a logic program~$\prog$ and an~\ASPDRUPE
  derivation~$\Pi$. Then, the \ASPDRUPE-Checker is correct,~i.e., it
  outputs \emph{Success} if and only if~$\Pi$ is an \ASPDRUPE proof
  for the inconsistency~of~$\prog$.
\end{corollary}
\longversion{%
\begin{proof}
  The statement follows immediately from
  Lemma~\ref{lem:invariant_rup_del}.
\end{proof}

}

\subsection{Extension to Optimization}
Next, we briefly mention how to verify cost optimization.  To this end, an
\emph{optimization rule} is an expression of the form
$\optimize l_{1}[w_1]$, where~$l_1$ is a literal.  Intuitively this
indicates that if an assignment satisfies~$\ba(\prog,\{l_1\})$, then
this results in costs~$w_1$.  Overall, one aims to minimize the
total costs, i.e., the goal is to deliver an answer set of minimal
total costs.
Therefore, if one wants to verify, whether a given answer set
candidate is indeed an answer set of minimal costs, we foresee the
following extension to~\ASPDRUPE, where such a proof consists of the
following two parts.
(i)~An answer set that shows a solution with costs~$c$ exists.
(ii)~An \ASPDRUPE proof that shows that the program restricted to
costs~$c-1$ is inconsistent.
Note that for disjunctive programs already the first part also needed to contain a second 
proof showing that indeed there cannot be an unfounded set for the provided answer set.
Further, it is not immediate, how this extends to unsatisfiable cores.
Hence, so far it only applies to progression based approaches.

\section{Integrating \ASPDRUPE Proofs into a Solver}

\begin{algorithm}
    \caption{{CDNL-\ASPDRUPE}: CDNL-ASP~\protect \cite[page 93]{GebserEtAl12} extended by proof logging}
    \label{algo:CDNL-ASP-RAT}
    \DontPrintSemicolon%
    \SetKwInOut{Input}{Input}
    \SetKwInOut{Output}{Output}
    \SetKwBlock{Loop}{loop}{end}

    \Input{A logic program $P$.}
    \Output{An answer set of $P$ or an \ASPDRUPE proof $\Pi$ certifying that $P$ has no answer set.}

    $A \coloneqq \emptyset$, $\nabla \coloneqq \emptyset$, $\dl \coloneqq 0$, \algonew{$\Pi \coloneqq \emptyset$}\;

    \Loop{%
        $(A, \nabla, \algonew{\Pi}) \coloneqq \algonew{\NogoodPropagation}(P, \nabla, A, \algonew{\Pi})$  \tcp*{deterministic cons.}

        \uIf(\tcp*[f]{conflict}){$\varepsilon \subseteq A$ \textbf{\textup{for some}} $\varepsilon \in \Delta_P \cup \nabla$}{%
            \If{$\max( \{ \dlevel(\sigma) \mid \sigma \in \varepsilon \} \cup \{ 0 \}) = 0$}{%
                \Return
                (INCONSISTENT,
                \algonew{$\Pi \seqconcat \langle (\proofstepadd, \Box, \dummyvar) \rangle$})
            }

            $(\delta, \dl) \coloneqq \ConflictAnalysis(\varepsilon, P, \nabla, A)$
            \tcp*{$\delta$ is RUP for $\Delta_P \cup \nabla$}

            $\nabla \coloneqq \nabla \cup \{ \delta \}$
            \tcp*{add conflict nogood}

            $\algonew{ \Pi \coloneqq \Pi \seqconcat \langle (\proofstepadd, \delta, \dummyvar) \rangle }$
            \tcp*{record nogood addition in proof}

            $A \coloneqq A \setminus \{ \sigma \in A \mid \dl < \dlevel(\sigma) \}$
            \tcp*{backjumping}
        }
        \uElseIf(\tcp*[f]{answer set found}){$A^\T \cup A^\F = \atom(P) \cup \bodies(P)$}{%
            \Return (CONSISTENT, $A^\T \cap \atom(P)$)\label{ln:aset}
        }
        \Else{%
            $\sigma_d \coloneqq \Select(P,\nabla,A)$ \tcp*{decision}
            $\dl \coloneqq \dl + 1$\;
            $\dlevel(\sigma_d) \coloneqq \dl$\;
            $A \coloneqq A \cup \{\sigma_d\}$
        }
    }
\end{algorithm}  %
\begin{algorithm}
    \caption{NogoodPropagation \protect \cite[page 101]{GebserEtAl12} adapted for proof logging}
    \label{algo:NogoodPropagation}
    \DontPrintSemicolon
    \SetKwInOut{Input}{Input}
    \SetKwInOut{Output}{Output}
    \SetKwBlock{Loop}{loop}{end}
    \SetKwFor{Let}{let}{in}{end}

    \Input{A logic program $P$, a set $\nabla$ of nogoods, an assignment $A$, and an \ASPDRUPE derivation $\Pi$.}
    \Output{An extended assignment, a set of nogoods, and an \ASPDRUPE derivation (possibly empty).}

    $U \coloneqq \emptyset$ \;

    \Loop{

        \Repeat{$\Sigma = \emptyset$}{%
            \If(\tcp*[f]{conflict}){$\delta \subseteq A$ \textbf{\textup{for some}} $\delta \in \Delta_P \cup \nabla$}{%
	        \algonew{$\Pi \coloneqq \Pi \seqconcat \langle \sigma \mid \sigma = (\proofstepcomp, \delta, \epsilon),  \delta\in\Delta_\prog^\leftarrow, \sigma \not\in \Pi \rangle \seqconcat$} \tcp*[f]{record confl.\ completion nogood}\label{line:conflictcompletion}
		\algonew{\makebox[3.5em]{}$\langle \sigma \mid \sigma = (\proofstepcompsup, \delta \setminus \{\T a\}, a), \delta\in\Delta_\prog^\rightarrow, \T a\in \delta, \sigma\not\in\Pi \rangle$}

                \Return $(A, \nabla, \algonew{\Pi})$%
            }%
            $\Sigma \coloneqq \{ \delta \in \Delta_P \cup \nabla \mid \delta \setminus A = \{ \lcompl{\sigma} \}, \sigma \not\in A \}$  \tcp*{unit-resulting nogoods}%
            \For(\tcp*[f]{record unit completion nogoods in proof}){$\delta \in \Sigma$}{%
                \algonew{$\Pi \coloneqq \Pi \seqconcat \langle \sigma \mid \sigma = (\proofstepcomp, \delta, \epsilon),  \delta\in\Delta_\prog^\leftarrow, \sigma \not\in \Pi \rangle \seqconcat$} \label{line:unitcompletion}%
		\algonew{\makebox[3.5em]{}$\langle \sigma \mid \sigma = (\proofstepcompsup, \delta \setminus \{\T a\}, a), \delta\in\Delta_\prog^\rightarrow, \T a\in \delta, \sigma\not\in\Pi \rangle$}
            }%
            \If{$\Sigma \neq \emptyset$}{%
                \Let{$\lcompl{\sigma} \in \delta$ \textbf{\textup{for some}} $\delta \in \Sigma$}{%
                    $\dlevel(\sigma) \coloneqq \max (\{ 0 \} \cup \{ \dlevel(\rho) \mid \rho \in \delta \setminus \{ \lcompl{\sigma} \} \})$ \;
                    $A \coloneqq A \cup \{\sigma\}$ \;
                }%
            }%
        }

        \lIf{$\lop(P) = \emptyset$}{%
            \Return $(A, \nabla, \algonew{\Pi})$
        }

        $U \coloneqq U \setminus A^\F$ \;

        \lIf{$U = \emptyset$}{%
            $U \coloneqq \UnfoundedSet(P, A)$
        }

        \lIf{$U = \emptyset$}{%
            \Return $(A, \nabla, \algonew{\Pi})$%
            \tcp*[f]{no unfounded set} %
        }

        \Let{$a_0 \in U$}{
            $\nabla \coloneqq \nabla \cup \{ \{ \T a_0 \} \cup \{ \F B \mid B \in \eb(\prog,U) \} \}$  \tcp*{add loop nogood}
            $\algonew{ \Pi \coloneqq \Pi \seqconcat \langle (\proofsteploop, \{ \T X \mid X \in U \}, a_0) \rangle }$  \tcp*{record loop in proof}
        }

    }
\end{algorithm}  %

In the following, we describe the CDNL-ASP algorithm for 
logic programs~$\prog$ that we use as a basis for our theoretical
model. Afterwards, we describe how \emph{proof logging} can be
integrated.  In other words, during the run of an ASP solver, we
immediately print the corresponding \ASPDRUPE rules that are needed
later for verification in case the ASP solver concludes that the
program is inconsistent.
A typical CDNL-based ASP solver (cf., Listing~\ref{algo:CDNL-ASP-RAT})
relies on unit propagation, since this is a rather simple and
efficient way of concluding consequences.  Thereby it keeps a
set~$\nabla$ of nogoods, a current assignment~$A$, and a decision
level~$\dl$.  In a loop it applies
\emph{\NogoodPropagation} \cite[page 101]{GebserEtAl12} consisting of
unit propagation and loop propagation
(using~$\emph{\UnfoundedSet}$ \cite[page 104]{GebserEtAl12}) whenever
suited.  Then, if there is some nogood that is not satisfied, either
the program is inconsistent (at decision level 0) or
$\emph{\ConflictAnalysis}$~\cite[page 108]{GebserEtAl12} triggers
backtracking to an earlier decision level, followed by the learning of
a conflict nogood~$\delta$.  Otherwise, if all nogoods
in~$\Delta_\prog\cup\nabla$ are satisfied and all the variables are
assigned, an answer set is found, and otherwise some free variable is
selected ($\emph{\Select}$).

Listings~\ref{algo:CDNL-ASP-RAT} and~\ref{algo:NogoodPropagation} contain a prototypical CDNL-based ASP solver that is extended by proof logging, where the changes for proof logging are highlighted in red.
\FIX{We use the element operator ($\in$) to determine whether an element is in a sequence,
and denote the concatenation of two proofs by the $\seqconcat$ operator
as follows:}
$\langle \sigma_1, \ldots, \sigma_i \rangle \seqconcat \langle
\sigma_{i+1}, \ldots, \sigma_n \rangle \eqdef \langle \sigma_1,
\ldots, \sigma_i, \sigma_{i+1}, \ldots, \sigma_n \rangle$.
The idea is to start with an empty~\ASPDRUPE derivation.
Whenever a new nogood, or loop nogood is learned and added to~$\nabla$ accordingly,
this results in an added addition or loop addition proof step, respectively.
\FIX{Note that in Listing~\ref{algo:NogoodPropagation} we add completion rule addition steps and completion support addition steps, whenever unit propagation (or conflicts) involve rules in~$\Delta_\prog^\leftarrow$ or~$\Delta_\prog^\rightarrow$, respectively.
In particular, Lines~\ref{line:conflictcompletion} and~\ref{line:unitcompletion} take care of adding involved parts of the completion to the proof (if needed) accordingly.}
At the end, when the ASP solver concludes inconsistency, the proof is returned including the empty nogood as last nogood.
Note that advanced techniques (see, e.g.,~\cite{GebserEtAl12}) like \emph{forgetting} of learned clauses and \emph{restarting} of the ASP solver can also be implemented using deletion rules with~\ASPDRUPE.
\future{%
Learned clauses in CDCL/CDNL have property RUP, which implies they also have property RAT.
However, many preprocessing techniques in SAT solving cannot be expressed with RUP clauses alone, but require addition/deletion of RAT clauses (cite drat paper for this?).
}
As it turns out, \emph{preprocessing} in ASP is less sophisticated as for SAT.
In the literature, CDNL-based ASP solvers often use preprocessing techniques~\cite{GebserEtAl08} similar to SAT solvers, i.e.,
SatElite-like~\cite{EenBiere05} operations as variable and nogood elimination.
For simple preprocessing operations restricted to variable and nogood elimination \ASPDRUPE suffices.
Note that if Clark's completion is exponential in the program size due to weight rules,
also propagators~\cite{AlvianoDodaroMaratea18} are supported. For details we refer to the implementation in Section~\ref{sec:impl}.

\future{In this algorithm, we only need RUP (not the more general RAT).
However, state-of-the-art ASP solvers augment this basic algorithm
to employ the same(?)  %
preprocessing and inprocessing techniques as modern SAT solvers.
Some of these techniques can be expressed more conveniently with RAT
(most likely RUP can also express everything, but: the solver might need
to spend extra time to construct the RUP proof; and the smallest equivalent RUP proof might
even be exponentially larger than the RAT proof -- see drat-trim paper \cite{WetzlerHeuleHunt14},
I think they show this property).
Modified nogood propagation is given in Algorithm~\ref{algo:NogoodPropagation}.}

\begin{example}[CDNL-\ASPDRUPE]
    \label{ex:cdnl}
    We continue the previous Example~\ref{ex:proof}
    and indicate a possible CDNL-\ASPDRUPE run on $\progex$
    that leads to the exemplary \ASPDRUPE proof given above.
    We use the notation $\T X @ \mathit{dl}$ ($\F X @ \mathit{dl})$ to indicate that
    $X$ was assigned true (false) at decision level $\mathit{dl}$.
    \begin{enumerate}
        \item
            Initially, nothing can be propagated.
        \item
            After the decision $\T a @ 1$,
            unit propagation derives only $\F \{\naf a, \naf e\} @ 1$.
        \item
            After the second decision $\F \{ c \} @ 2$,
            we eventually derive $\T a @ 2$ and $\T b @ 2$ by unit propagation.
            Thus we discover the unfounded set $U = \{a,b\} \in \progloops(\progex)$,
            and add the loop nogood $\lambda(a, U)$.
        \item
            The loop nogood immediately leads to a conflict,
            and conflict analysis learns the nogood $\{ \F \{ c \}, \T a \}$.
        \item
            We backtrack to decision level 1, and after propagation, make the decision~$\T e @ 2$.
            We arrive at another conflict, and learn $\{ \T e, \F \{ \naf a, \naf e \} \}$.
        \item
            After backtracking, a conflict appears at decision level 1,
            and we learn $\{ \T a \}$.
        \item
            We backtrack to decision level 0, and decide on $\T e @ 1$.
            After arriving at a conflict almost immediately, we learn $\{ \T e \}$.
        \item
            We backtrack once more, and finally arrive at a conflict at decision level~0,
            returning INCONSISTENT along with an \ASPDRUPE proof.
            Note that the proof given in Example~\ref{ex:proof} is slightly simplified
            in that we only include those steps of types $\proofstepcomp$ and $\proofstepcompsup$
            that are actually used. %
    \end{enumerate}%
\end{example}

\subsection{Implementation of~\ASPDRUPE in wasp solver}\label{sec:impl}
We provide an implementation of~\ASPDRUPE within the wasp~\cite{AlvianoEtAl15} solver
that is available on github\footnote{The repository can be found at~\url{https://github.com/alviano/wasp/tree/unsatproof}.}.
The solver prints a proof for inconsistency in the file \textsf{proof.log} if the solver gets passed the program options \textsf{-\,-disable-simplifications -\,-proof-log=proof.log}. 
Actually wasp prints an \ASPDRUPE derivation also in the positive case of consistency. 
This derivation can still be used to verify whether the nogoods learned by the solver are correct. %
Currently, proof logging is restricted to normal programs \FIX{and we do not yet support
recursive weight rules due to discrepancy among different semantics as discussed in the preliminaries.}
Moreover, we had to introduce a normalized form because of several (in-processing) simplifications that would otherwise require major refactoring to isolate.
Just to mention one of these simplifications, a rule of the form $a \hsep \ell_1, \ldots, \ell_n, \neg a$ is replaced by the integrity constraint $\hsep \ell_1, \ldots, \ell_n, \neg a$.
While these simplifications are required to achieve efficient computation, they alter the completion of the input program.
Therefore, wasp cannot log in the proof the auxiliary atoms required to keep the completion compact.
The problem is circumvented by introducing a normalized form such that the completion can be compactly computed without introducing any auxiliary atoms.

A program~$\prog$ is in
\emph{short-body normalized form} if for each atom~$a\in\at(\prog)$
either~$\Card{\ebs(\prog,a)}\leq1$, or
for any body~$B\in\ebs(\prog,a)$, we have $\Card{B}\leq 1$. 
Any normal program can be transformed into short-body normalized form in linear time 
by introducing a linear
number of auxiliary atoms (in the program size). %
This normalized form allows us to get rid of auxiliary variables for bodies~$B\in\bodies(\prog)$, i.e., we can set~$\Delta_\prog^\compdef=\emptyset$, and replace~$\T B$ in~$\Delta_\prog^\leftarrow$ by~$\ba(\prog,B)$, and~$\F B$ in~$\Delta_\prog^\rightarrow$ by~$\lcompl{\ba(\prog,B)}$.
For simplification and increased readability of a compact resulting proof log,
we further do not use neither completion rule addition, nor completion support addition types. 
Instead, we assume that the checker is aware of the completion from the beginning.
\FIX{In this respect}, we have to observe that for weight rules, completion might be exponential in the program size.
Therefore, we require that the checker is equipped with a propagator~\cite{AlvianoDodaroMaratea18} for drawing conclusions
by unit-propagation on parts of the completion associated with weight rules.
We provide an implementation of such a checker tool as well\footnote{Both the checker tool and a tool for bringing normal logic programs in short-body normalized form can be found at~\url{https://github.com/alviano/python/tree/master/asp-proof}.}. %

\longversion{
\section{Extensions for Disjunctive ASP \& Optimization}\label{sec:extensions}
\todoi{maybe rewrite concerning disjunction?}
In the following, we describe the changes that are necessary to handle
disjunctive programs. Further, we provide an outlook on cost
optimization. %

\subsection{Disjunctive Programs}
We can extend \ASPDRUPE to a disjunctive program~$\prog$ as follows.
A \emph{proof step for disjunctive programs} is a quadruple
$(t, \delta, a, A)$, where
$t \in \{ \proofsteploop, \proofstepadd, \proofstepext, \proofstepdel,
\proofstepunfound \}$, $\delta$ is an assignment, $a$ is an atom
or~$\dummyvar$ if unused, and~$A$ is also an assignment or the empty
assignment if unused.  Then, a proof step $\sigma_i$ of an \ASPDRUPE
derivation for disjunctive programs is either (i)~an \ASPDRUPE
derivation proof step for normal programs where triples are extended
to quadruples with last position~$A=\emptyset$ or (ii)~an
\emph{unfounded set addition}, which is a proof step of the form
$(\proofstepunfound, \{ \T a_1, \ldots, \T a_k \},$ $\epsilon, A)$ for
some unfounded set $U = \{ a_1, \ldots, a_k \}$ and assignment~$A$
such that~$A^\T\cap U\neq\emptyset$. %
This step represents the addition of~$A$ as nogood, as~$A$ can be
excluded~\cite{Faber05}.
We adapt the definition of~$\nabla_i(\Pi)$ for~$1\leq i\leq n$:
\begin{align*}
    \nabla_{i}(\Pi) &\coloneqq
    \begin{cases}
        \ldots
            & \ldots,
            \\
            \nabla_{i-1}(\Pi) \msum \{ A \} &
            \text{if $\sigma_i = (\proofstepunfound, \{\T a_1, \ldots, \T a_k\}, \dummyvar, A)$.}
    \end{cases}
\end{align*}

\noindent Further, $\Xi_\prog\eqdef \{A \mid U \text{ is an unfounded set for } A \text{ in } \prog, A^\T\cap U \neq\emptyset\}$.

\future{Loops as before.
But now there are additional unfounded sets that don't correspond to loops.
This unfounded set check is NP-complete and usually implemented as post-check once a model has been found.
No unfounded set => we have an answer set (this case is not interesting for us here).
Unfounded set => add nogood to exclude it and continue.

So what we can do for \ASPDRUPE is to add a line $\mathtt{u}~l_1~\ldots~l_n~\mathtt{0}$ to the proof format
which corresponds to these more complicated unfounded sets.
The unfounded set itself is the yes-certificate for the NP-complete unfounded set check and thus can be checked in polynomial time. (exact algorithm not given here?)
But we should still keep the loop nogood check as a simpler/faster check.}

\begin{lemma}[Invariants for Disjunctive Programs]\label{lem:invariant_disj_del} 
  Let $\prog$ be a disjunctive program,
  $\langle \sigma_1, \ldots, \sigma_n \rangle$ a finite \ASPDRUPE
  derivation for program $P$, and $D_{i}$ be the accumulated set of
  nogoods introduced by the extension rules in $\sigma_i$ for all
  $i \in \{1, \ldots, k\}$. Then, it holds that
  $\Delta_\prog \cup \Xi_\prog \cup D_i \models \nabla_i(\Pi) \text{ for all } i \in \{0, \ldots, n\}$
\end{lemma}
\begin{proof}[Proof (Idea)]
The proof for the missing unfounded set addition step case is similar to the loop addition case in the proof of Lemma~\ref{lem:invariant_rup_del}. 
\end{proof}

\begin{theorem}[Correctness for Disjunctive Programs]
  A disjunctive logic program $P$ is inconsistent if and only if there exists an \ASPDRUPE proof for $P$.
\end{theorem}
\begin{proof}[Proof (Idea)]
  We proceed analogously to the proof of Theorem~\ref{ref:normalasp},
  but additionally rely on~$\Xi_\prog$ and on the characterization of
  answer sets for disjunctive programs~\cite{Faber05}. In particular,
  $A^\T\cap\at(\prog)$ is an answer set of~$\prog$ if and only
  if~$A\models \Delta_\prog\cup\Lambda_\prog\cup\Xi_\prog$, since
  every loop~$U\in\lop(\prog)$ is an unfounded set for~$\lambda(a,U)$
  with $a\in U$. 
\end{proof}

Note that the checker in Listing~\ref{algo:proof-checker} can be extended easily.
The algorithm {CDNL-\ASPDRUPE} of Listing~\ref{algo:CDNL-ASP-RAT} can
also be extended to disjunctive programs.  Therefore, it needs an
additional unfounded set check in Line~\ref{ln:aset}, where, if an
unfounded set~$U$ for assignment~$A$ is found, $A$ is added to
$\nabla$, and the unfounded set addition step for~$U$ and~$A$ is added
to~$\Pi$ and the solver continues.  Analogously, the checker in
Listing~\ref{algo:proof-checker} can be extended accordingly, since
the claimed unfounded sets can be verified in polynomial time.
}

\section{\ASPDRUPE Implementational Specifications}\label{sec:implspec}
Next, we discuss the specific format description of~\ASPDRUPE that we
think can be commonly used in ASP solvers.
To this end, we assume a program~$\prog$ and a set~$\{\hat B\mid B\in \bodies(\prog)\}$ of fresh atoms, where we have one fresh atom for each induced body in~$\prog$. 
Further, let~$\vm: (\at(\prog)\cup\{\hat B\mid B\in\bodies(\prog)\}) \rightarrow \Nat$ be an injective mapping of atoms~$a\in\at(\prog) \cup\{\hat B\mid B\in\bodies(\prog)\}$ to a unique positive integer~$n\geq 1$ such that~$\vm(a) \eqdef n$, and $a = \vm^{-1}(n)$.
\FIX{Note that for atoms~$a\in\at(\prog)$ this can be (partly) already provided by the input format.}
However, for technical reasons, we assume such a mapping also for atom~$\hat B$, where $B\in\bodies(\prog)$, as these integers will then correspond to fresh atoms.
We define in the following an SModels-like~\cite{lparse} output format of strings for a given program~$\prog$, which is ready for the checker to parse. 
\FIX{Actually, the output format is ``line-based'', i.e., it is even ASPIF-like~\cite{GebserEtAl16}. However, \ASPDRUPE still supports ASP only, and not ASP solving with theory reasoning.}
To this end, let the \emph{truth value mapping}~$\outn$ map a variable assignment~$l$ to an integer different from 0, where a positive integer represents an atom and a negative integer a negated atom.
\begin{align*}
    \outn(l) &\coloneqq
    \begin{cases}
    	\vm(X),
            & \text{if $l=\T Y$ and~$X=Y$ is an atom or~$X=\hat Y$ for~$Y\in\bodies(\prog)$,}
            \\
    	-\vm(X),
            & \text{if $l=\F Y$ and~$X=Y$ is an atom or~$X=\hat Y$ for~$Y\in\bodies(\prog)$.}
    \end{cases}
\end{align*}
Then, the~\emph{\ASPDRUPE output format} is a sequence~$\zeta=\langle s_1,\ldots,s_j \rangle$ of strings,
where each element in the sequence corresponds to one rule in an~\ASPDRUPE derivation and is terminated by character~``0''.
Each part of an element in the sequence is separated by a white space (\text{\textvisiblespace}). We indicate
other strings constants   by~$'string'$.
Then, element~$s_i$ of the sequence~$\zeta$ for~$1\leq i \leq j$ is of the following form.

\begin{itemize}
        \item
        A \emph{body definition string} is of the form $'\proofstepbody'\text{\textvisiblespace}{}b_1\text{\textvisiblespace}{}n_1\text{\textvisiblespace}{}\ldots\text{\textvisiblespace}{}n_k\text{\textvisiblespace}{}'0'$ such that $b_1$, $n_1$, $\ldots$, $n_k\in\Nat$.
	Further, we require that $\{\outn^{-1}(n_1), \ldots, \outn^{-1}(n_k)\}=\ba(\prog,B)$ for some~$B \in \bodies(\prog)$.
        Finally, for $s_i$ the string corresponds to the proof step~$(\proofstepext, \{\T B\}, {\hat B})$, where~${\hat B}=\vm^{-1}(b_1)$.
        The technical purpose of~$s_i$ is to specify the fresh body variable~$\hat B$ for a body~$B\in\bodies(\prog)$. %
        \item
        An \emph{addition string} is of the form $'\proofstepadd'\text{\textvisiblespace}{}n_1\text{\textvisiblespace}{}\ldots\text{\textvisiblespace}{}n_k\text{\textvisiblespace}{}'0'$ such that~$n_1,\ldots,n_k\in\Nat$, which corresponds to proof step~$(\proofstepadd, \{\outn^{-1}(n_1), \ldots, \outn^{-1}(n_k)\}, \dummyvar)$.
    \item
	A \emph{completion rule addition string} is of the form $'\proofstepcomp'\text{\textvisiblespace}{}b_1\text{\textvisiblespace}{}n_1\text{\textvisiblespace}{}\ldots\text{\textvisiblespace}{}n_k\text{\textvisiblespace}{}'0'$ such that~$b_1,n_1,\ldots,n_k\in\Nat$ and~${\hat B}=\vm^{-1}(b_1)$, which relates to proof step~$(\proofstepcomp, \{\F \vm^{-1}(n_1), \ldots, \F \vm^{-1}(n_k), \T B\},\dummyvar)$.
    \item
	A \emph{completion support addition string} is of the form $'\proofstepcompsup'\text{\textvisiblespace}{}n_1\text{\textvisiblespace}{}b_1\text{\textvisiblespace}{}\ldots\text{\textvisiblespace}{}b_k\text{\textvisiblespace}{}'0'$ such that~$n_0,b_1,\ldots,$ $b_k\in\Nat$ and~$\hat B_i=\vm^{-1}(b_i)$ for~$i\in\{1,\ldots,k\}$, 
	which corresponds to proof step~$(\proofstepcompsup, \{\F B_1,\ldots,$ $\F B_k\}, \vm^{-1}(n_0))$.
    \item
        An \emph{extension string} is of the form $'\proofstepext'\text{\textvisiblespace}{}n_0\text{\textvisiblespace}{}n_1\text{\textvisiblespace}{}\ldots\text{\textvisiblespace}{}n_k\text{\textvisiblespace}{}'0'$ such that~$n_0$, $n_1$, $\ldots$, $n_k\in\Nat$, which corresponds to proof step~$(\proofstepext, \{\outn^{-1}(n_1), \ldots, \outn^{-1}(n_k)\}, \vm^{-1}(n_0))$.
        \item
        A \emph{deletion string} is of the form $'\proofstepdel'\text{\textvisiblespace}{}n_1\text{\textvisiblespace}{}\ldots\text{\textvisiblespace}{}n_k\text{\textvisiblespace}{}'0'$ such that~$n_1,\ldots,n_k\in\Nat$, which corresponds to proof step~$(\proofstepdel, \{\outn^{-1}(n_1), \ldots, \outn^{-1}(n_k)\},\dummyvar)$.
        \item
        A \emph{loop addition string} is of the form $'\proofsteploop'\text{\textvisiblespace}{}n_1\text{\textvisiblespace}{}\ldots\text{\textvisiblespace}{}n_k\text{\textvisiblespace}{}'0'$ such that~$n_1$, $\ldots$, $n_k\in\Nat$, which corresponds to proof step~$(\proofsteploop, \{\T\vm^{-1}(n_1), \ldots, \T\vm^{-1}(n_k)\},\vm^{-1}(n_1))$.
        \item An \emph{unfounded set addition string} is of the form $'\proofstepunfound'\text{\textvisiblespace}{}k\text{\textvisiblespace}{}n_1\text{\textvisiblespace}{}\ldots\text{\textvisiblespace}{}n_k\text{\textvisiblespace}{}o_1\text{\textvisiblespace}{}\ldots\text{\textvisiblespace}{}o_m\text{\textvisiblespace}{}'0'$ such that~$n_1,$ $\ldots,n_k,o_1,\ldots,o_m\in\Nat$, which then corresponds to proof step $(\proofstepunfound, \{\T\vm^{-1}(n_1), \ldots, \T\vm^{-1}(n_k)$ $\},\dummyvar, \{\outn^{-1}(o_1), \ldots, \outn^{-1}(o_m)\})$.
\end{itemize}

Next, we define how to obtain the \ASPDRUPE output format~$\zeta=\smof(\Pi)$ of a given \ASPDRUPE derivation~$\Pi=\langle\sigma_1,\ldots,\sigma_n\rangle$.
To this end, we  define~$s=\smof(\sigma_i)$, which transforms a proof step~$\sigma_i$ into a string~$s$ for~$1\leq i\leq n$, by slight abuse of notation.
\begin{align*}
    \smof(\sigma_i) &\coloneqq
    \begin{cases}
	\mathsf{'\proofstepadd'\text{\textvisiblespace}{}\outn(l_1)\text{\textvisiblespace}{}\ldots\text{\textvisiblespace}{}\outn(l_k)\text{\textvisiblespace}{}'0'},
            & \text{if $\sigma_i = (\proofstepadd, \{l_1, \ldots, l_k\}, \dummyvar)$,}
            \\
	\mathsf{'\proofstepcomp'\text{\textvisiblespace}{}\vm(\hat B)\text{\textvisiblespace}{}\vm(a_1)\text{\textvisiblespace}{}\ldots\text{\textvisiblespace}{}\vm(l_k)\text{\textvisiblespace}{}'0'},
            & \text{if $\sigma_i = (\proofstepcomp, \{\F a_1, \ldots, \F a_k, \T B\}, \dummyvar)$,}
            \\
	\mathsf{'\proofstepcompsup'\text{\textvisiblespace}{}\vm(a)\text{\textvisiblespace}{}\vm(B_1)\text{\textvisiblespace}{}\ldots\text{\textvisiblespace}{}\vm(B_k)\text{\textvisiblespace}{}'0'},
            & \text{if $\sigma_i = (\proofstepcompsup, \{\F B_1, \ldots, \F B_k\}, a)$,}
            \\
    	\mathsf{'\proofstepext'\text{\textvisiblespace}{}\vm(a)\text{\textvisiblespace}{}\outn(l_1)\text{\textvisiblespace}{}\ldots\text{\textvisiblespace}{}\outn(l_k)\text{\textvisiblespace}{}'0'},
            & \text{if $\sigma_i = (\proofstepext, \{l_1, \ldots, l_k\}, a)$,}
            \\
    	\mathsf{'\proofstepdel'\text{\textvisiblespace}{}\outn(l_1)\text{\textvisiblespace}{}\ldots\text{\textvisiblespace}{}\outn(l_k)\text{\textvisiblespace}{}'0'},
            & \text{if $\sigma_i = (\proofstepdel, \{l_1, \ldots, l_k\}, \dummyvar)$,}
	    \\
	\mathsf{'\proofsteploop'\text{\textvisiblespace}{}\vm(a_1)\text{\textvisiblespace}{}\ldots\text{\textvisiblespace}{}\vm(a_k)\text{\textvisiblespace}{}'0'},
            & \text{if $\sigma_i = (\proofsteploop, \{\T a_1, \ldots, \T a_k\}, a_1)$,}
            \\
	\mathsf{'\proofstepunfound'\text{\textvisiblespace}{}k\text{\textvisiblespace}{}\vm(a_1)\text{\textvisiblespace}{}\ldots\text{\textvisiblespace}{}\vm(a_k)}, & \text{if $\sigma_i = (\proofstepunfound, \{\T a_1, \ldots, \T a_k\}, \dummyvar, $}
\\
	\quad\quad\;\mathsf{\outn(l_1)\text{\textvisiblespace}{}\ldots\text{\textvisiblespace}{}\outn(l_m)\text{\textvisiblespace}{}'0'},
	    & \quad\qquad\qquad\text{$\{l_1, \ldots, l_m\})$.}
    \end{cases}
\end{align*}
Since fresh body atoms require introduction using extension proof steps
in advance, we assume~$\bodies(\prog)=\{B_1,\ldots,B_q\}$, where~$B_i=\{l_{i,1}, \ldots l_{i,\Card{B_i}}\}$ for~$1\leq i\leq q$.
Finally, let~$\smof(\Pi) \eqdef \langle'\proofstepbody'\text{\textvisiblespace}{}\vm(\hat{B}_1)\text{\textvisiblespace}{}\outn(l_{1,1})\text{\textvisiblespace}{}\ldots\text{\textvisiblespace}{}\outn(l_{1,\Card{B_1}})\text{\textvisiblespace}{}'0'\rangle \seqconcat \ldots \seqconcat \langle'\proofstepbody'\text{\textvisiblespace}{}\vm(\hat{B}_q)\text{\textvisiblespace}{}\outn(l_{q,1})\text{\textvisiblespace}{}\ldots\text{\textvisiblespace}{}\outn(l_{q,\Card{B_q}})\text{\textvisiblespace}{}'0' \rangle  \seqconcat  \smof(\sigma_1) \seqconcat \ldots \seqconcat \smof(\sigma_n)$.
As a simplification, one can leave out additional, unused body definition strings.

\begin{example}
    Consider the \ASPDRUPE proof~$\Pi$ for inconsistency of~$\prog$ from Example~\ref{ex:proof} and
    assume the dictionary in the program input assigns
    $\vm(a) = 1$, $\vm(b) = 2$, $\vm(c) = 3$, $\vm(d) = 4$, and $\vm(e) = 5$.
    We extend this to the necessary bodies:
    $\vm( \{c\} ) = 6$,
    $\vm( \{\naf c\} ) = 7$,
    $\vm( \{\naf d\} ) = 8$,
    $\vm( \{a, d\} ) = 9$,
    $\vm( \{b, d\} ) = 10$,
    $\vm( \{c, d\} ) = 11$,
    $\vm( \{\naf a, \naf e\} ) = 12$, and
    $\vm( \{c, \naf e\} ) = 13$.
    Figure~\ref{fig:output} corresponds to~$\smof(\Pi)$.
    Note that we use body definitions for required body variables.
\end{example}

\begin{figure}    \noindent%
    \begin{minipage}{.3\textwidth}
        \begin{aspdrupe}
        b 6 3 0
        b 7 -3 0
        b 8 -4 0
        b 9 1 4 0
        b 10 2 4 0
        b 11 3 4 0
        b 12 -1 -5 0
        b 13 3 -5 0
        \end{aspdrupe}
    \end{minipage}\hfill%
    \begin{minipage}{.3\textwidth}
        \begin{aspdrupe}
        s 1 10 6 0
	c 8 3 0
	c 7 4 0
        c 9 2 0
        s 5 13 12 0
        l 1 2 0
        a -6 1 0
        s 3 8 0
        s 4 7 0
        \end{aspdrupe}
    \end{minipage}\hfill%
    \begin{minipage}{.3\textwidth}
        \begin{aspdrupe}
        s 2 9 11 0
	a 5 -12 0
	c 13 5 0
        a 1 0
        c 10 1 0
        c 6 1 0
        a 5 0
        c 12 5 0
        a 0
        \end{aspdrupe}
    \end{minipage}
    \caption{\ASPDRUPE output format~$\smof(\Pi)$ of the proof~$\Pi$ of Example~~\ref{ex:proof}.}
    \label{fig:output}
\end{figure}

\section{Conclusion \& Future Work}
ASP solvers are highly-tuned decision procedures that are widely applied in academia and industry.
In this paper, we considered  how to ensure that if an ASP solver outputs that a program has no answer set then the solver is indeed right.
Similar to unsat certificates in SAT solvers, we propose an approach that  augments the inconsistency answer of an ASP solver with a certificate of inconsistency.
This approach allows the use of unverified, efficient ASP solvers while guaranteeing that a particular run of an ASP solver has been correct.

To this end, we developed a new proof format called~\ASPDRUPE. It allows several types of proof steps: \emph{(RUP) addition} that models nogood
learning, \emph{completion rule addition} and \emph{completion support addition} for adding completion rules on demand, \emph{deletion} that models nogood forgetting,
\emph{extension} that allows to infer new definitions and %
\emph{loop addition} that %
adds nogoods to forbid assignments that do not correspond to answer
sets.
\ASPDRUPE supports formula simplification methods that can be obtained by learning entailed nogoods, nogood deletion as well as extended resolution.
We established that \ASPDRUPE is sound and complete for logic programs
and can be used effectively,~i.e., a program $P$ is inconsistent if
and only if a \ASPDRUPE proof of $P$ exists and  that we can check an \ASPDRUPE proof in polynomial time of the proof length. %
Further, we demonstrated how to augment CDNL-based solvers with proof logging.
\FIX{It is in our interest for future work to continue this line of research.
Potential next steps include the study of theory reasoning towards covering the full ASPIF intermediate format.}

Finally, we would like to say a few words about RAT-style proofs.
The combination of nogood deletion, loop, and RAT addition results in an inconsistent proof system in which we can infer a conflict although a non-tight program is consistent.
This stems from the situation that clauses that are RAT with respect to $\nabla$ are not necessarily RAT with respect to $\Delta_\prog \cup \Lambda_\prog$.
Although it was recently shown that extended resolution simulates DRAT~\cite{KieslEtAl18},
we are unaware whether  \ASPDRUPE can be extended to RAT such that each rule application can be checked in polynomial time, which we believe is an interesting question for future work.

\shortversion{
\bibliography{asp_rat_refs}
}

\longversion{
\bibliography{asp_rat_refs}}

\appendix
\longversion
{
\clearpage
\section{Additional Resources}

\appendix

\subsection{Additional Examples}

}

\end{document}